\documentclass[11pt]{article}
\usepackage{path,graphicx}
\usepackage{a4,amssymb, dsfont}
\usepackage{authblk}
\usepackage{amsmath}
\usepackage{cite}
\usepackage{multirow}
\usepackage{lineno}
\usepackage{mathtools}
\usepackage{color}
\usepackage{here}
\usepackage{stackrel}
\usepackage[]{algorithm2e}
\usepackage[margin=1.1in]{geometry}
\usepackage[utf8]{inputenc}
\usepackage[english]{babel}
\usepackage{float}
\usepackage{rotating}
\usepackage{path,graphicx}
\usepackage{a4,amssymb,dsfont}
\usepackage{microtype}
\usepackage{mathrsfs}
\usepackage{stackrel}
\usepackage{csquotes}
\usepackage[utf8]{inputenc}
\usepackage[]{algorithm2e}
\usepackage{scalerel}
\usepackage{xcolor}
\usepackage{fancyhdr}
\usepackage[justification=centering]{caption}
\newtheorem{definition}{Definition}[section]
\newtheorem{lemma}[definition]{Lemma}
\newtheorem{theorem}[definition]{Theorem}

\newtheorem{observation}[definition]{Observation}

\newcommand{\qed}{\hfill $\square$\smallskip}

\newcommand{\argmax}{\mathop{\mathrm{argmax}}}

\title{Spread of Influence in Graphs\thanks{A preliminary version of this paper~\cite{zehmakan2019tight} appeared in 13th International Conference on Language and Automata Theory and Applications.}}
\author{Ahad N. Zehmakan\thanks{Corresponding author; Email Address: abdolahad.noori@inf.ethz.ch, Postal Address: CAB G 39.3, Institute of Theoretical Computer Science, ETH Z\"urich, Universit\"atstrasse 6, CH-8092 Z\"urich.}}
\affil{Department of Computer Science, ETH Zurich}
\providecommand{\keywords}[1]{\textbf{\textit{Key Words:}} #1}
\date{} 
\begin{document}
\maketitle
\begin{abstract}
Consider a graph $G$ and an initial configuration where each node is black or white. Assume that in each round all nodes simultaneously update their color based on a predefined rule. One can think of graph $G$ as a social network, where each black/white node represents an individual who holds a positive/negative opinion regarding a particular topic. In the $r$-threshold (resp. $\alpha$-threshold) model, a node becomes black if at least $r$ of its neighbors (resp. $\alpha$ fraction of its neighbors) are black, and white otherwise. The $r$-monotone (resp. $\alpha$-monotone) model is the same as the $r$-threshold (resp. $\alpha$-threshold) model, except that a black node remains black forever.

What is the number of rounds that the process needs to stabilize? How many nodes must be black initially so that black color takes over or survives? Our main goal in the present paper is to address these two questions.
\end{abstract}
\section{Introduction}
In 1991, an elaborate rumor circled the African American community. The rumor was that a brand of soda, Tropical Fantasy Soda Pop, was made by the Ku-Klux-Klan. Not only that, but the soda was made with a special formula that made black men sterile. Of course, it was just a rumor and one that turned out to be false. Nevertheless, sales of the soda dropped by 70\% following the rumor, and people even began attacking the company's delivery trucks. How do rumors like this spread, and how can they become so powerful that they prompt action even when there is no merit to the claim?

Traditionally, the rumors used to spread simply by word of mouth, but these days, rumors have advanced technology to help them travel. With the development of social network services, such as Facebook, Twitter, and Instagram, the public has been provided with more extensive, fast, and real-time information access and exchange platforms. For example on April 23, 2013, a fake tweet from a hacked Associated Press account claimed that explosions at the White House had injured Barack Obama. That lone tweet caused instability on world financial markets, and the Standard and Poor's 500 Index lost \$130 billion in a short period.\footnote{See~\cite{coombe1997demonic,fisher2013syrian} for more details on these two examples and also other examples.}

Countless real-world examples of this kind demonstrate that the spread of (mis-)information can have a significant impact on various subjects, such as economy, defense, fashion, even personal affairs. Furthermore, the wide use of Internet virtually has sped up the dissemination of information and influence through communities. Therefore in the last few years, there has been a growing interest in the following fundamental questions:   
\begin{displayquote}
\centering \textit{How does a rumor spread?}\\ \textit{How do people form their opinion?}
\\ \textit{How do a community's members influence each other?}
\end{displayquote}
Researchers from a wide spectrum of fields, such as political and social science, computer science, statistical physics, and mathematics, have devoted a substantial amount of effort to address these questions. From a theoretical perspective, it is natural to introduce and study mathematical models which mimic rumor spreading and opinion formation in a community. In the real world, they are too complex to be explained in purely mathematical terms. However, the main idea is to comprehend their general principles and make crude approximations at discovering certain essential aspects of them which are otherwise totally hidden by the complexity of the full phenomenon. In particular, there are recurring patterns that one can try to capture mathematically.

Some of the most well-established models are DK model, majority dynamics, push/pull protocols, Galam model, energy model, and SIR model. We are interested in the following very basic model, which actually captures some of the aforementioned models.
\paragraph{Basic Model.}Consider a graph $G$ and an initial configuration where each node is either black or white. Assume that in discrete-time rounds, all nodes simultaneously update their color based on a predefined rule. 

One can think of graph $G$ as a social network, the graph of relationships and interactions within a group of individuals. That is, each node represents an individual and we add an edge between two nodes if the corresponding individuals are, for instance, friends. Furthermore, black and white stand for the state of an individual, for example, informed/uninformed about a rumor or positive/negative regarding a reform proposal. For the suitable choices of the updating rule, this model can potentially mimic the rumor spreading and opinion forming processes such as the diffusion of technological innovations in a community or the rise of a political movement in an unstable society. For example, assume that a college student starts using a cellphone if a certain fraction or number of his/her classmates have one, or an individual adopts the positive opinion regarding a reform proposal if a majority of his/her friends are positive, and negative otherwise. Motivated by such examples, we consider the following five updating rules.
\paragraph{$r$-Threshold Model.}In the \emph{$r$-threshold model} for some fixed positive integer $r$, a node becomes black if it has at least $r$ black neighbors and white otherwise. See Figure~\ref{r example} for an example. 
\begin{figure}[!ht]
\centering
\includegraphics[scale=0.8]{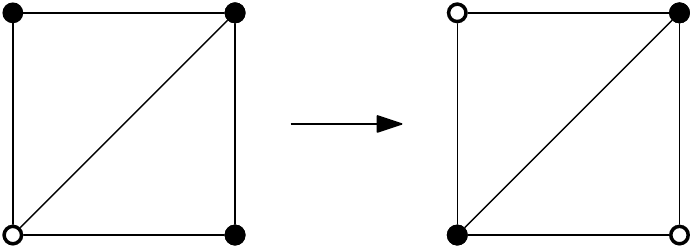}
\caption{One round of the $r$-threshold model for $r=2$.\label{r example}}
\end{figure}
\paragraph{$\alpha$-Threshold Model.}In the \emph{$\alpha$-threshold model} for some constant $0<\alpha<1$, a node becomes black if at least $\alpha$ fraction of its neighbors are black and white otherwise.\footnote{There should not be any confusion between the $\alpha$-threshold and $r$-threshold model since r is an integer value larger than equal to $1$ and $\alpha$ is smaller than $1$.} See Figure~\ref{alpha example} for an example.

\begin{figure}[!ht]
\centering
\includegraphics[scale=0.8]{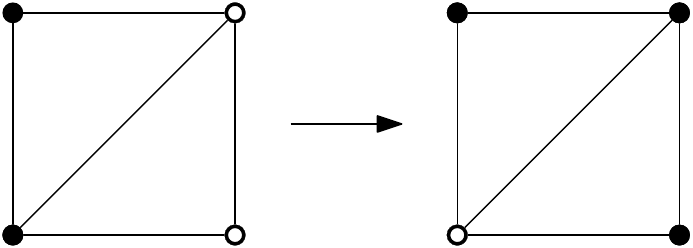}
\caption{One round of the $\alpha$-threshold model for $\alpha=1/2$.\label{alpha example}}
\end{figure}
\paragraph{$r$-Monotone and $\alpha$-Monotone Model.}The \emph{$r$-monotone model} (resp. \emph{$\alpha$-monotone model}) is the same as the $r$-threshold model (resp. $\alpha$-threshold model), except that a black node remains black forever.\footnote{These two models sometimes are known as bootstrap percolation in this literature.} For example, if we consider the $r$-monotone (resp. $\alpha$-monotone) model instead of the $r$-threshold (resp. $\alpha$-threshold) model in Figure~\ref{r example} (resp. Figure~\ref{alpha example}), the graph becomes fully black in the next round.

We notice that on a $d$-regular graph the $r$-threshold (resp. $r$-monotone) model is the same as the $\alpha$-threshold (resp. $\alpha$-monotone) model for $\alpha=r/d$.
\paragraph{Majority Model.}We also define the \emph{majority model} where each node updates its color to the most frequent color in its neighborhood and keeps its current color in case of a tie. See Figure~\ref{majority example} for an example. 

\begin{figure}[!ht]
\centering
\includegraphics[scale=0.8]{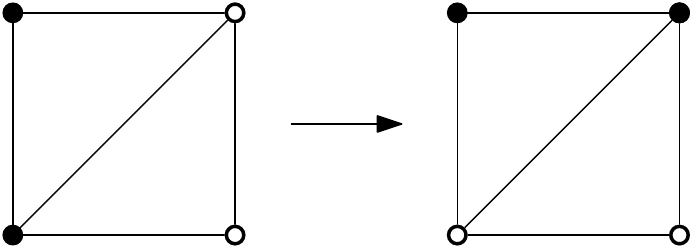}
\caption{One round of the majority model.\label{majority example}}
\end{figure}
We observe that the majority model is basically the same as the $\alpha$-threshold model for $\alpha=1/2$, up to the tie-breaking rule. The examples in Figures~\ref{alpha example} and~\ref{majority example} should illustrate the difference in the tie-breaking rule. 

\paragraph{Stabilization Time and Periodicity.}The updating rule in all of these five models is deterministic and there are $2^{n}$ possible configurations for an $n$-node graph $G$. Thus, by starting from an arbitrary configuration on graph $G$, the process must eventually reach a cycle of configurations. The length of this cycle and the number of rounds that the process needs to reach it are called the \emph{periodicity} and the \emph{stabilization time} of the process, respectively. In the $r$-monotone and $\alpha$-monotone model the periodicity is one and the stabilization time is upper-bounded by $\mathcal{O}(n)$, regardless of the initial configuration. As our first main contribution, we provide the tighter bound of $\mathcal{O}(n/\delta)$ on the stabilization time of the $r$-monotone model, where $\delta$ denotes the minimum degree in $G$. Furthermore for the $r$-threshold, $\alpha$-threshold, and majority model, $2^n$ is a trivial upper bound on both periodicity and stabilization time of the process. However, interestingly Goles and Olivos~\cite{goles1980periodic} proved that the periodicity is always one or two. Furthermore, Fogelman,
Goles, and Weisbuch~\cite{fogelman1983transient} showed that the stabilization time is bounded by $\mathcal{O}\left(m\right)$, where $m$ is the number of edges in $G$.
 
\paragraph{Dynamo.}Arguably, the most well-studied question in this context is to determine the number of nodes which must be black initially so that black color \textit{takes over} (i.e., the whole graph becomes black eventually). A node set $D$ is said to be a \emph{dynamic monopoly}, or shortly \emph{dynamo}, if black color takes over once all nodes in $D$ are black. See Figure~\ref{dynamo example} for an example. Thus, the above question is essentially asking for the minimum size of a dynamo. In our examples from above, this can be interpreted as the minimum number of individuals who must support an innovation or opinion to make it go viral and take over the whole community. 
\begin{figure}[!ht]
\centering
\includegraphics[scale=0.8]{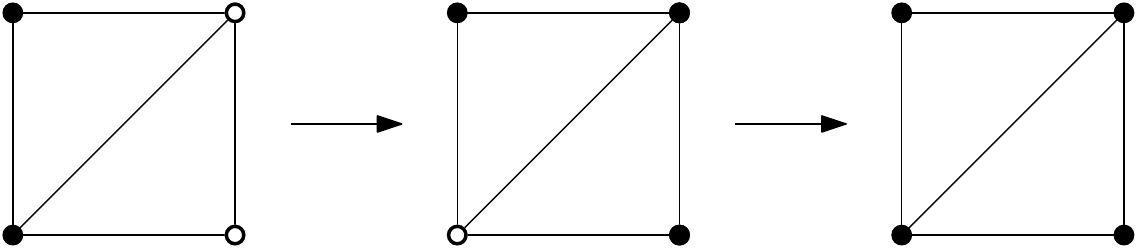}
\caption{A dynamo in the $\alpha$-threshold model for $\alpha=1/2$.\label{dynamo example}}
\end{figure}

Though the concept of a dynamo had been studied before, e.g. by Balogh
and Pete~\cite{balogh1998random} and Schonmann~\cite{schonmann1992behavior}, it was first introduced formally by Peleg~\cite{peleg1997local}. This initiated an extensive study of dynamos on different classes of graphs, including the $d$-dimensional torus~\cite{hambardzumyan2017polynomial,balister2010random,morrison2018extremal,jeger2019dynamic,gartner2017color,gartner2017majority,zehmakan2018two}, hypercube~\cite{morrison2018extremal,balogh2006bootstrap,balogh2009majority}, Erd\H{o}s–R\'{e}nyi random graph~\cite{janson2012bootstrap,feige2017contagious,zehmakan2019opinion,zehmakan2019target}, random regular graphs~\cite{coja2015contagious,guggiola2015minimal,gartner2018majority}, power-law random graphs~\cite{amini2014bootstrap,abdullah2014phase}, and so on.

Here, we do not limit ourselves to a specific graph structure and our main goal is to characterize the behavior of the above five models on general graphs. In particular, we aim to establish several lower and upper bounds on the minimum size of a dynamo in terms of the number of nodes and edges in the underlying graph. See for example Table~\ref{general:Table 1}, where we present lower and upper bounds in terms of the number of nodes and of course the parameters $\alpha$ and $r$. A lower bound $LB(n,\alpha,r)$ implies that in an arbitrary graph $G$ any dynamo is of size at least $LB$. An upper bound $UB(n,\alpha,r)$ implies that in an arbitrary graph $G$ there is a dynamo of size at most $UB$. Some of these bounds are known by prior work, which are marked with the corresponding references, and our results are also marked with the according theorems, observations, and lemmas. Some of these bounds are quite straightforward. For instance in the $r$-monotone and $r$-threshold model, $r$ is an obvious lower bound on the minimum size of a dynamo, and it is tight since the complete graph $K_n$ has a dynamo of size $r$ (see Lemma~\ref{general:lemma} for a proof). However, some of the bounds are much more involved and require novel ideas. All bounds given in Table~\ref{general:Table 1} are tight up to an additive constant, except for certain ranges of $\alpha$ and the upper bound in the majority model. We say a lower/upper bound is tight if there is a graph $G$ for which the minimum size of a dynamo is equal to it.

\begin{table}[t]
\centering
\begin{tabular}{ |l|c|c| } 
 \hline
Model& Lower Bound &Upper Bound\\
 \hline
 $\alpha$-monotone   &1~(Obs.~\ref{general:thm1}) & $2\alpha n$~\cite{garbe2018contagious}\\
 \hline
 $\alpha$-threshold $\ \alpha>3/4$&   $2\alpha\sqrt{n}-1$~\cite{berger2001dynamic}  & $n$\\
 \hline
 $\alpha$-threshold $\ \alpha\le3/4$&   $1$~(Lem.~\ref{general:obs1})  & $n$\\
\hline
majority model & 1~\cite{berger2001dynamic} & $n$\\
\hline
 $r$-monotone   & $r$   &$rn/(r+1)$~\cite{reichman2012new}\\
 \hline
 $r$-threshold $\ r\ge 2$&   $r$  &$n$ \\
 \hline
 $r$-threshold  $\ r=1$ & 1~(Thm.~\ref{general:thm4}) & 2~(Thm.~\ref{general:thm4})\\
 \hline
\end{tabular}
\caption{The minimum size of a dynamo.\label{general:Table 1}}
\end{table}

The techniques that we use to prove the bounds in Table~\ref{general:Table 1} are fairly standard and straightforward (some new and some inspired from prior work). The upper bounds are built on the probabilistic method and some greedy algorithms. For the lower bounds, we define a suitable potential function, like the number of edges whose endpoints have different colors in a configuration. Then, careful analysis of the behavior of such a potential function during the process allows us to establish lower bounds on the size of a dynamo. To prove the tightness of our results, we provide explicit graph constructions for which the minimum size of a dynamo matches our bounds.

Similar to dynamo, we also consider the minimum size of a monotone dynamo, for the $\alpha$-threshold, $r$-threshold, and majority model, and present several tight bounds. A node set $D$ is a monotone dynamo if it is a dynamo and takes over the whole graph monotonically. Notice that any dynamo is monotone by definition in the $r$-monotone and $\alpha$-monotone model. Therefore, all bounds from Table~\ref{general:Table 1} hold also for the minimum size of a monotone dynamo in these two models.

All the above bounds are a function of the number of nodes in the underlying graph (and the model parameters $r$ and $\alpha$). We also consider the number of edges. In particular, we prove that in the $r$-threshold model any dynamo is of size at least $2\left(n-\frac{m}{r}\right)$. We will argue that even though this bound is provided in a general framework, it is tight for special classes of graphs such as the $d$-dimensional torus. The main idea behind the proof of the above bound is to show that in some particular settings, any lower bound on the minimum size of a dynamo in bipartite graphs holds actually for all graphs. Thus, we need to consider only bipartite graphs, which can be handled more easily.

A simple observation is that by adding an edge to a graph, the minimum size of a dynamo in the $r$-monotone and $r$-threshold model does not increase (for a formal argument, please see Section~\ref{general:dynamo}). Thus, if one keeps adding edges to a graph, eventually it will have a dynamo of minimum possible size, i.e., $r$. Thus, it would be interesting to ask for the degree-based density conditions which ensure that a graph $G$ has a dynamo of size $r$. This was studied for the $r$-monotone model, with regard to the minimum degree, by Freund, Poloczek, and Reichman~\cite{freund2018contagious}. They proved that if the minimum degree $\delta\left(G\right)$ is at least $\lceil (r-1)n/r\rceil$, then there is a dynamo of size $r$ in $G$. Gunderson~\cite{gunderson2017minimum} showed that the statement holds even for $\delta \ge n/2+r$, and this is tight up to an additive constant. We study the same problem for the threshold variant and prove that if $\delta\ge \frac{n}{2}+r$ then the graph includes $\Omega\left(n^r\right)$ dynamos of size $r$. Note that this statement is stronger than Gunderson's result in the following two ways. Firstly, we prove that there is a dynamo of size $r$ in the $r$-threshold model, which implies that there is a dynamo of size $r$ in the $r$-monotone model. Moreover, we show that there is not only one but also $\Omega\left(n^r\right)$ of such dynamos. It is worth to stress that our proof is substantially shorter and simpler.

\paragraph{Robust and Eternal Sets.} So far we focused on the minimum number of black nodes which can make the whole graph black. One might relax this and ask for the minimum number of black nodes which guarantee the survival of black color forever; that is, there is at least one black node in all upcoming rounds. To address this question, we introduce and study the concept of an eternal set. A node set $S$ is said to be \emph{eternal} if there will be at least one black node in all upcoming rounds once all nodes in $S$ are black. We also consider the concept of a robust set. A node set $S$ is said to be \emph{robust} if all nodes remain black forever once $S$ is fully black.

Similar to dynamo, we provide tight bounds on the minimum size of a robust and an eternal set, as a function of $n$, $r$, and $\alpha$. In the $r$-monotone and $\alpha$-monotone model, a black node stays unchanged; thus, the minimum size of a robust/eternal set is equal to one for any graph. However, the situation is a bit more involved in the $r$-threshold, $\alpha$-threshold, and majority model. Building on several new proof techniques, we bound the minimum size of a robust/eternal set for these three models. See Tables~\ref{general:Table 2Robust} and~\ref{general:Table 2Eternal} for a summary of our results, where $x=0$ and $x=1$ respectively for odd and even $n$.
All these bounds are tight, except the bound of $2\alpha n+1/\alpha$.

\begin{table}
\centering
\begin{tabular}{|l|c|c|}
\hline
Model& Lower Bound & Upper Bound\\
\hline
$\alpha$-threshold $\ \alpha\le\frac{1}{2}$&   $\lceil \frac{1}{1-\alpha}\rceil$~(Thm.~\ref{general:thm7})  & $2\alpha n+1/\alpha$~(Thm.~\ref{general:thm7})\\
  \hline
$\alpha$-threshold $\ \alpha>\frac{1}{2}$&   $\lceil \frac{1}{1-\alpha}\rceil$~(Thm.~\ref{general:thm7})  & $n$~(Thm.~\ref{general:thm7})\\
  \hline
  majority model& $2$ & $\lfloor n/2\rfloor+1$~(Thm.~\ref{general:majoritythm1})\\
    \hline
$r$-threshold $\ r=1$&   $2$  &$2$\\
 \hline
$r$-threshold $\ r\ge 2$ & $r+1$ & $n$\\
 \hline
\end{tabular}
\caption{The minimum size of a robust set.\label{general:Table 2Robust}}
\end{table}

\begin{table}
\centering
\begin{tabular}{|l|c|c|c|c|}
 \hline
Model & Lower Bound & Upper Bound\\
 \hline
$\alpha$-threshold $\ \alpha\le\frac{1}{2}$& $1$ & $2\alpha n+1/\alpha$\\
  \hline
$\alpha$-threshold $\ \alpha>\frac{1}{2}$ & $1$ & $n$\\
  \hline
  majority model & $1$ & $\lfloor n/2\rfloor+1$~(Thm.~\ref{general:majoritythm1})\\
    \hline
$r$-threshold $\ r=1$& $1$~(Thm.~\ref{general:thm8}) & $1$~(Thm.~\ref{general:thm8})\\
 \hline
$r$-threshold $\ r=2$& $2$~(Thm.~\ref{general:thm8}) & $\frac{n}{1+x}$~(Thm.~\ref{general:thm8})\\
  \hline
$r$-threshold $\ r\ge 3$ & $r$~(Thm.~\ref{general:thm8}) & $n$~(Thm.~\ref{general:thm8})\\
 \hline
\end{tabular}
\caption{The minimum size of an eternal set.\label{general:Table 2Eternal}}
\end{table}

To give the reader a better understanding, let us explain some of the bounds in the above tables, which are quite straightforward to prove. The minimum size of an eternal set in the $\alpha$-threshold model on a graph $G$ is trivially at least 1 for any $0<\alpha<1$. This trivial bound is tight since in the $\alpha$-threshold model on the star graph $S_n$, the internal node is an eternal set of size 1, regardless of the value of $\alpha$. Furthermore, it is easy to observe that the upper bound of $\lfloor n/2\rfloor+1$ is tight for the complete graph $K_n$. Consider the majority model on an arbitrary graph $G=(V,E)$. To prove the above upper bound, we consider a partition of the node set into two subsets $V_1$ and $V_2$ such that $\lfloor n/2\rfloor-1\le |V_2|\le|V_1|\le \lfloor n/2\rfloor+1$ and the number of edges in between is minimized. In Theorem~\ref{general:majoritythm1}, we show that each node in $V_1$ has at least half of its neighbors in $V_1$, which implies that $V_1$ is a robust set of size at most $\lfloor n/2\rfloor+1$.

\paragraph{Outline.} First we provide some basic definitions and introduce our models formally in Section~\ref{preliminaries}. Then, we study the stabilization time in our five models in Section~\ref{general:stabilization time}, where we overview prior work and provide new results. Finally as our main contribution, we present several tight bounds, in terms of the number of nodes and edges, on the minimum size of a dynamo in Section~\ref{general:dynamo} and on the minimum size of a robust set and an eternal set in Section~\ref{general:robust}.

\section{Preliminaries}
\label{preliminaries}
\paragraph{Graph Definitions.}Let $G=\left(V,E\right)$ be a graph. We always assume that $n$ and $m$ denote the number of nodes and edges in $G$, respectively. Furthermore, it is assumed that $G$ is connected; otherwise we point it out explicitly. 

For a node $v\in V$, $\Gamma\left(v\right):=\{u\in V: \{u,v\} \in E\}$ is the \emph{neighborhood} of $v$. For a set $S\subset V$, we define $\Gamma\left(S\right):=\bigcup_{v\in S}\Gamma\left(v\right)$ and $\Gamma_S\left(v\right):=\Gamma\left(v\right)\cap S$. Moreover, $d\left(v\right):=|\Gamma\left(v\right)|$ is the \emph{degree} of $v$ and $d_S\left(v\right):=|\Gamma_S\left(v\right)|$. We also define $\delta\left(G\right)$ to be the minimum degree in graph $G$. (To simplify notation, we sometimes shortly write $\delta$ where $G$ is clear form the context). Moreover, for two nodes $u,v \in V$, let $d\left(u,v\right)$ denote the length of the shortest path between $v,u$ in terms of the number of edges, which is called the \emph{distance} between $v$ and $u$ (for a node $v$, we define $d\left(v,v\right)=0$). Furthermore, for a node set $S\subset V$ we define the \emph{edge boundary} of $S$ to be $\partial\left(S\right):=\{\{u,v\}:v\in S,	 u\in V\setminus S\}$. 

\paragraph{Models.}For a graph $G=\left(V,E\right)$, a \emph{configuration} is a function $\mathcal{C}:V\rightarrow\{b,w\}$, where $b$ and $w$ represent black and white, respectively. If $\mathcal{C}$ is a constant function, we call it a \emph{monochromatic} configuration and \emph{bichromatic} otherwise. For a node $v\in V$, the set $\Gamma_a^{\mathcal{C}}\left(v\right):=\{u\in \Gamma\left(v\right):\mathcal{C}\left(u\right)=a\}$ includes the neighbors of $v$ which have color $a\in\{b,w\}$ in configuration $\mathcal{C}$. For a set $S\subset V$, we define $\Gamma_a^{\mathcal{C}}\left(S\right):=\bigcup_{v\in S} \Gamma_a^{\mathcal{C}}\left(v\right)$. Furthermore, we write $\mathcal{C}|_S=a$ if $\mathcal{C}\left(u\right)=a$ for every $u\in S$.

Assume that we are given an initial configuration $\mathcal{C}_0$ on a graph $G=\left(V,E\right)$. In a model $M$, $\mathcal{C}_t\left(v\right)$ for $t\ge 1$ and $v\in V$, which is the color of node $v$ in the $t$-th configuration, is determined based on a predefined updating rule. We are mainly interested in the five following models, where $\mathcal{C}_t\left(v\right)$ is defined by a deterministic updating rule as a function of $\mathcal{C}_{t-1}\left(u\right)$ for $u\in \Gamma\left(v\right)\cup \{v\}$. 

In the \emph{$r$-threshold model} for some integer $r\ge 1$
\[
\mathcal{C}_t\left(v\right) =\left\{\begin{array}{lll}b &\mbox{if $|\Gamma^{\mathcal{C}_{t-1}}_b\left(v\right)|\ge r$} \\
w,&\mbox{otherwise}
\end{array}\right.
\]
for any $t\ge 1$ and node $v\in V$. In words, node $v$ is black in the $t$-th configuration if and only if it has at least $r$ black neighbors in the previous configuration.

The \emph{$r$-monotone model} is the same as the $r$-threshold model except that a black node remains black forever. More formally, for a node $v$ we have $\mathcal{C}_t\left(v\right)=b$ if and only if $|\Gamma_b^{\mathcal{C}_{t-1}}\left(v\right)|\ge r$ or $\mathcal{C}_{t-1}\left(v\right)=b$. 

In the \emph{$\alpha$-threshold model} for some constant $0<\alpha<1$
\[
\mathcal{C}_t\left(v\right) =\left\{\begin{array}{lll}b &\mbox{if $|\Gamma^{\mathcal{C}_{t-1}}_b\left(v\right)|\ge \alpha d\left(v\right)$} \\
w,&\mbox{otherwise}
\end{array}\right.
\]
for any $t\ge 1$ and node $v\in V$.

The \emph{$\alpha$-monotone model} is the same as the $\alpha$-threshold model except that a black node remains black forever. More formally, for a node $v$ we have $\mathcal{C}_t\left(v\right)=b$ if and only if $|\Gamma_b^{\mathcal{C}_{t-1}}\left(v\right)|\ge \alpha d\left(v\right)$ or $\mathcal{C}_{t-1}\left(v\right)=b$. 

In the \emph{majority model} for any $t\ge 1$ and node $v\in V$ 
\[
\mathcal{C}_t\left(v\right) =\left\{\begin{array}{lll}\mathcal{C}_{t-1}\left(v\right), &\mbox{if $|\Gamma^{\mathcal{C}_{t-1}}_b\left(v\right)|=|\Gamma^{\mathcal{C}_{t-1}}_w\left(v\right)|$}, \\
\argmax_{a\in\{b,w\}} |\Gamma^{\mathcal{C}_{t-1}}_a\left(v\right)|,&\mbox{otherwise.}
\end{array}\right.
\]

In each of the above models, we define $B_t$ and $W_t$ for $t\ge 0$ to be the set of black and white nodes in the $t$-th configuration. Furthermore, let $b_t:=|B_t|$ and $w_t:=|W_t|$ denote the number of black and white nodes in the $t$-the configuration.

\paragraph{Assumptions.}We assume that $r$ and $\alpha$ are fixed while we let $n$ tend to infinity. Furthermore, note that if $d\left(v\right)<r$ for a node $v$ and it is initially white, it never becomes black in the $r$-threshold and $r$-monotone model. Thus, we always assume that $r\le \delta\left(G\right)$.

Consider a model $M$ on a graph $G=\left(V,E\right)$. We define the three following concepts.
\begin{definition}[dynamo]
A node set $S\subseteq V$ is a \emph{dynamo} whenever the following holds: If all nodes in $S$ are black in some configuration, then black color takes over.  
\end{definition}
\begin{definition}[eternal set]
A node set $S\subseteq V$ is an \emph{eternal set} whenever the following holds: If all nodes in $S$ are black in some configuration, then black color survives.  
\end{definition}
\begin{definition}[robust set]
A non-empty node set $S\subseteq V$ is a \emph{robust set} whenever the following holds: If all nodes in $S$ are black in some configuration, then they will remain black forever.  
\end{definition}
If a set $S$ is a robust set, then it is an eternal set by definition, but not necessarily the other way around. 

A node set $S\subseteq V$ is a \emph{monotone dynamo} if it is a dynamo and takes over monotonically. That is, if the set of black nodes is equal to $S$ in some configuration, then this set keeps growing until black color takes over. More formally, if $B_t=S$ for some $t\ge 0$, then $B_{t'}\subseteq B_{t'+1}$ for any $t'\ge t$ and $B_{t^{\prime\prime}}=V$ for some $t^{\prime\prime}\ge t$. In the $r$-monotone and $\alpha$-monotone model, any dynamo is a monotone dynamo, but this is not true in the other three models.

In the $r$-monotone and $\alpha$-monotone model, any non-empty set $S$ is robust since a black node remains black forever. However, a non-empty node set $S$ is a robust set in the $r$-threshold model (resp. $\alpha$-threshold model) if and only if for each node $v\in S$, $d_S\left(v\right)\ge r$ (resp. $d_S\left(v\right)\ge \alpha d\left(v\right)$).

Now, we set up some notations for the minimum size of a (monotone) dynamo, a robust set, and an eternal set. For a graph $G=\left(V,E\right)$, we define the following notations:
\begin{itemize}
\item $MD\left(G,r\right):=$ minimum size of a dynamo in $r$-monotone
\item $\overleftarrow{MD}\left(G,r\right):=$ minimum size of a dynamo in $r$-threshold
\item $MR\left(G,r\right):=$ minimum size of a robust set in $r$-monotone
\item $\overleftarrow{MR}\left(G,r\right):=$ minimum size of a robust set in $r$-threshold
\item $ME\left(G,r\right):=$ minimum size of an eternal set in $r$-monotone.
\item $\overleftarrow{ME}\left(G,r\right):=$ minimum size of an eternal set in $r$-threshold.
\end{itemize}
We analogously define $MD\left(G,\alpha\right)$, $MR\left(G,\alpha\right)$, $ME\left(G,\alpha\right)$ for the $\alpha$-monotone model, $\overleftarrow{MD}\left(G,\alpha\right)$, $\overleftarrow{MR}\left(G,\alpha\right)$, and $\overleftarrow{ME}\left(G,\alpha\right)$ for the $\alpha$-threshold model, and $\overleftarrow{MD}\left(G,maj\right)$, $\overleftarrow{MR}\left(G,maj\right)$, and $\overleftarrow{ME}\left(G,maj\right)$ for the majority model. We should mention that the reason for the use of the backward arrow $\leftarrow$ in the above notations for the $\alpha$-threshold, $r$-threshold, and majority models is that a white node which has become black, might turn white again.

Finally, we define $\overleftarrow{MD}_{mon}\left(G,r\right)$, $\overleftarrow{MD}_{mon}\left(G,\alpha\right)$, and $\overleftarrow{MD}_{mon}\left(G,maj\right)$ for the minimum size of a monotone dynamo respectively in the $r$-threshold, $\alpha$-threshold, and majority model. (Note that we skipped the $r$-monotone and $\alpha$-monotone model since each dynamo is a monotone dynamo in these two models.) 
\section{Stabilization Time and Periodicity}
\label{general:stabilization time}
In this section, we discuss several bounds on the stabilization time and periodicity in each of our five models. Let us start with the $r$-monotone model. The periodicity is always one. A trivial upper bound on the stabilization time is $n-r$ because if $b_0<r$ then no node changes its color and if $b_0\ge r$ in each round at least one white node becomes black until we reach a fixed configuration. This simple bound is actually tight. Consider graph $H=(V_H=\{v_1,\cdots,v_n\},E_H)$, where each node $v_i$ for $i\ge r+1$ is adjacent to $v_1,\cdots,v_{r-1}$ and $v_{i-1}$. See Figure~\ref{general:GeneralBoundsFig1} for an example with $r=3$. If in the $r$-monotone model on $H$ initially nodes $v_1,\cdots,v_r$ are black, then in the $i$-th round node $v_{i+r}$ becomes black until the whole graph is black. This takes $n-r$ rounds.
\begin{figure}[!ht]
\centering
\includegraphics[scale=1.0]{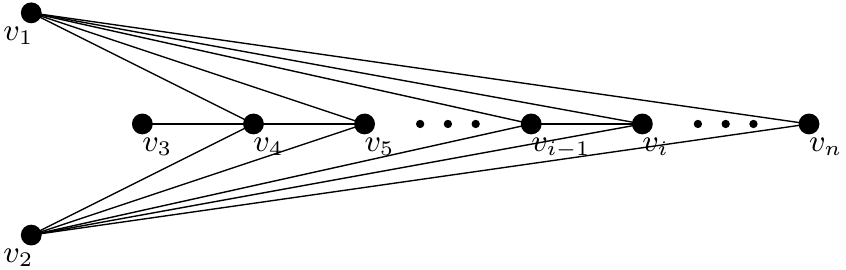}
\caption{The construction of graph $H$ for $r=3$.}
\label{general:GeneralBoundsFig1}
\end{figure}

It is also quite straightforward to prove that the stabilization time of the $r$-monotone model on a graph $G$ is bounded by the maximum length of a simple path in $G$. Consider an arbitrary initial configuration $\mathcal{C}_0$ for which the stabilization time is equal to $T$ for some $T\ge 0$. We show that $G$ has a simple path of length $T$. Let $B'_0=B_0$ and define $B'_t=B_t\setminus B_{t-1}$ to be the set of nodes which become black in the $t$-th round. Each node in $B'_t$ has at least one neighbor in $B'_{t-1}$; otherwise it would have become black before round $t$. This immediately implies that $G$ includes a simple path of length $T$ along sets $B'_t$ for $0\le t\le T$. Therefore, our claim follows. This simple bound is also tight, up to an additive constant. Again consider the construction of graph $H=(V_H,E_H)$ which we described above. (See Figure~\ref{general:GeneralBoundsFig1}) Since $v_1,\cdots, v_r$ create an independent set, a path can include at most two nodes among $v_1,\cdots, v_r$. Thus, Graph $H$ does not have any simple path of length $n-r+2$ or larger. On the other hand, if initially only $v_1,\cdots, v_r$ are black the process takes $n-r$ rounds. 

The above bound is not very interesting for dense graphs since they include long simple paths. For example, if $\delta\ge n/2$ for a graph $G$, then it has a simple path of length $n-1$ by Dirac's theorem. However, the stabilization time for very dense graphs such as the complete graph $K_n$ can be bounded by a constant. To address this observation, we prove in Theorem~\ref{general:stabiliztion:delta} an upper bound of form $\mathcal{O}(\frac{rn}{\delta})$, which is equivalent to $\mathcal{O}(\frac{n}{\delta})$ in our setting since we assume that $r$ is a constant.

\begin{theorem}
\label{general:stabiliztion:delta}
In the $r$-monotone model on a graph $G$, the stabilization time is bounded by $\mathcal{O}\left(\frac{rn}{\delta}\right)$.
\end{theorem}
\begin{proof}
Assume that there is an initial configuration $\mathcal{C}_0$ for which the stabilization time is at least $18 r n/\delta$. Similar to above, we define $B'_0=B_0$ and $B'_t=B_t\setminus B_{t-1}$ for $t\ge 1$. Arrange the first $18 r n/\delta$ $B'_t$s into $6rn/\delta$ triples $(B'_{3j+1},B'_{3j+2},B'_{3j+3})$ for $0\le j\le 6rn/\delta-1$. We call a triple \emph{heavy} if it contains more than $\delta/3$ nodes, and \emph{light} otherwise. We observe that at least $3rn/\delta$ of the triples are light because otherwise the graph includes more than $(3rn/\delta)\cdot (\delta/3)=rn$ nodes.

We call an edge $\{v,u\}$ for $v\in B'_{t_1}$ and $u\in B'_{t_2}$ a \emph{jump edge} if $|t_1-t_2|>1$. Each node $w\in B'_t$ is adjacent to at most $r-1$ nodes in $\bigcup_{t'=0}^{t-2}B'_{t'}$ because otherwise it becomes black before the $t$-th round (which is in contradiction with the definition of $B'_t$). Thus, the number of jump edges is upper-bounded by $(r-1)n<rn$.

So far we proved that there are at least $3rn/\delta$ light triples and less than $rn$ jump edges. Therefore, there must exist a light triple $(B'_{3j'+1}, B'_{3j'+2}, B'_{3j'+3})$ which is incident to less than $(2rn)/(3rn/\delta)=2\delta/3$ jump edges. Let $v$ be an arbitrary node in $B'_{3j'+2}$; it is adjacent to at most $\delta/3$ nodes in the triple (since it is a light triple) and it is incident to less than $2\delta/3$ jump edges. This implies that $v$'s degree is less than $2\delta/3+\delta/3=\delta$, which is a contradiction. Therefore, the stabilization time is less than $18 r n/\delta$ for any initial configuration. \qed
\end{proof}

In the $\alpha$-monotone model, trivially the periodicity is one, as in the $r$-monotone model, but there is not much known about the stabilization time, which would be interesting to study in future research. 
\begin{theorem}[\cite{goles1980periodic}]
\label{Goles} In the $r$-threshold, $\alpha$-threshold, and majority model, the periodicity is always one or two.
\end{theorem}
A monochromatic configuration is an example of periodicity one. For periodicity two, consider a complete bipartite graph $G=\left(V_1\cup V_2,E\right)$, where $|V_1|=|V_2|$. If initially all nodes in $V_1$ are black, the process keeps alternating between two configurations in all of the three models, regardless of the choice of $r$ and $\alpha$.
\begin{theorem}[\cite{fogelman1983transient}]
\label{fogelman}
In the $r$-threshold, $\alpha$-threshold, and majority model on a graph $G$, the stabilization time is bounded by $\mathcal{O}\left(m\right)$.
\end{theorem}
\section{Dynamo}
\label{general:dynamo}
Here, we aim to prove tight bounds on the minimum size of a dynamo in all our five models. In Section~\ref{general:numberofnodes}, we prove several bounds in terms of the number of nodes in the underlying graph, which are summarized in Table~\ref{general:Table 1}. Then, we present some other bounds with regard to the number of edges of the graph in Section~\ref{general:nymberofedges} 
\subsection{Number of Nodes}
\label{general:numberofnodes}
As a warm-up, let us bound the minimum size of a dynamo for some specific classes of graphs in Lemma~\ref{general:lemma}, which actually come in handy several times later for arguing the tightness of our bounds.
\begin{lemma}
\label{general:lemma}
For the complete graph $K_n$
\[
MD\left(K_n,r\right)=\overleftarrow{MD}\left(K_n,r\right)=r\ \ \text{and}\ \ MD\left(K_n,\alpha\right)\ge\lceil\alpha n\rceil-1.
\]
Furthermore, for an $r$-regular graph $G=\left(V,E\right)$ and $r\ge 2$
\[
\overleftarrow{MD}\left(G,r\right)=n.
\]
\end{lemma}
\begin{proof}
Firstly, $r\le MD\left(K_n,r\right)$ because by starting from a configuration with less than $r$ black nodes in the $r$-monotone model, all nodes keep their color unchanged forever. Secondly, $\overleftarrow{MD}\left(K_n,r\right)\le r$ because from a configuration with $r$ black nodes in the $r$-threshold model, in the next round all the $n-r$ white nodes become black and after one more round all nodes will be black because $n-r$ is at least $r+1$ (recall that we assume that $r$ is fixed while $n$ tends to infinity). By these two statements and the fact that $MD\left(K_n,r\right)\le\overleftarrow{MD}\left(K_n,r\right)$ (this is true since a dynamo in the $r$-threshold model is also a dynamo in the $r$-monotone model), we have $MD\left(K_n,r\right)=\overleftarrow{MD}\left(K_n,r\right)=r$. In the $\alpha$-monotone model on $K_n$, by starting with less than $\lceil \alpha \left(n-1\right)\rceil$ black nodes, all white nodes stay white forever, which implies that $MD\left(K_n,\alpha\right)\ge \lceil \alpha \left(n-1\right)\rceil\ge \lceil\alpha n\rceil-1$. (The interested reader might try to find the exact value of $MD\left(K_n,\alpha\right)$ as a small exercise.) 

Consider an arbitrary configuration with at least one white node, say $v$, in the $r$-threshold model on an $r$-regular graph $G$. In the next round all nodes in $\Gamma\left(v\right)$ will be white. Thus, by starting from any configuration except a fully black one, black color never takes over. This implies that $\overleftarrow{MD}\left(G,r\right)=n$. (We exclude $r=1$ because a $1$-regular graph is disconnected for large $n$.) \qed
\end{proof}
\subsubsection{$\mathbf{\alpha}$-Monotone Model}
For the $\alpha$-monotone model on a graph $G$, Chang~\cite{chang2011triggering} proved that $MD(G,\alpha)\le \left(2\sqrt{2}+3\right)\alpha n<5.83\alpha n$. The tighter bound of $4.92\alpha n$ was provided in~\cite{chang2012triggering}. Finally, Garbe, Mycroft, and McDowell~\cite{garbe2018contagious} proved that $MD(G,\alpha)\le 2\alpha n$. We do not know whether this bound is tight or not.

Furthermore, trivially we have that $1\le MD\left(G,\alpha\right)$. We argue in Observation~\ref{general:thm1} that this bound is actually best possible.
\begin{observation}
\label{general:thm1}
For the star graph $S_n$, $MD(S_n,\alpha)=1$.
\end{observation}
Consider the $\alpha$-monotone model on $S_n$. Assume that initially the internal node is black. In the next round, all leaves will become black since they only see black color in their neighborhood and the internal node remains black.

\subsubsection{$\mathbf{\alpha}$-Threshold and Majority Model}
In this section, we bound the minimum size of a dynamo in the $\alpha$-threshold and majority model. We consider these two models together since some of our arguments can be applied to both of them. 
\paragraph{Lower Bounds.}For $\alpha\le 1/2$, we prove that the trivial lower bound of $1\le \overleftarrow{MD}(G,\alpha)$ is tight.
\begin{lemma}
\label{general:obs1}
$\overleftarrow{MD}\left(C_n,\alpha\right)=1$ for $\alpha\le \frac{1}{2}$ and odd $n$. 
\end{lemma}
\begin{proof}
Consider an odd cycle $v_1,v_2,\cdots,v_{2k+1},v_1$ for $n=2k+1$ and the $\alpha$-threshold model. Assume that in the initial configuration $\mathcal{C}_0$ there is at least one black node, say $v_1$. In configuration $\mathcal{C}_1$, nodes $v_{2}$ and $v_{2k+1}$ are both black because $\alpha\le 1/2$. With a simple inductive argument, after $k$ rounds two adjacent nodes $v_{k+1}$ and $v_{k+2}$ will be black. In the next round, they both stay black and nodes $v_{k}$ and $v_{k+3}$ become black as well. Again with an inductive argument, after at most $k$ more rounds all nodes will be black. \qed
\end{proof}

Let us consider the majority model before moving to the case of $\alpha>1/2$. Consider a clique of size $\sqrt{n}$ and attach $\sqrt{n}-1$ different leaves to each of the nodes in this clique. This simple construction provides us with an $n$-node graph which has a dynamo of size $\sqrt{n}$ in the majority model. Note that if all nodes in the clique are black, in the next round all nodes become black. Can we do better? The answer is yes. Berger~\cite{berger2001dynamic}, surprisingly, proved that there exist arbitrarily large graphs which have dynamos of constant size in the majority model. Therefore, the trivial lower bound of $1\le \overleftarrow{MD}(G,maj)$ is tight, up to an additive constant.

Can we strengthen Berger's result by providing arbitrarily large graphs which have dynamos of constant size in the $\alpha$-threshold model for any $\alpha>1/2$? The answer is negative. Berger~\cite{berger2001dynamic} also proved that $2\alpha\sqrt{n}-1\le \overleftarrow{MD}(G,\alpha)$ for $\alpha>3/4$. The answer is not known for $1/2<\alpha\le 3/4$. 

Adapting the above construction, one can show that the lower bound of $2\alpha\sqrt{n}-1$ for $\alpha>3/4$ is tight, up to a multiplicative constant. We provide $n$-node graphs with dynamos of size $k=\sqrt{\alpha n/(1-\alpha)}$ for $\alpha> 3/4$ (actually this construction works also for $\alpha\ge 1/2$). Consider a clique of size $k$ and attach $\frac{n}{k}-1$ leaves to each of its nodes. The resulting graph has $k+k\left(\frac{n}{k}-1\right)=n$ nodes. Consider the initial configuration $\mathcal{C}_0$ in which the clique is fully black and all other nodes are white. In $\mathcal{C}_1$, all the leaves become black because their neighborhood is fully black in $\mathcal{C}_0$. Furthermore, each node $v$ in the clique remains black since it has $k-1$ black neighbors and $k-1\ge \alpha d\left(v\right)$ for $\alpha>1/2$, which we prove below. Thus, it has a dynamo of size $k$ in the $\alpha$-threshold model.
\begin{align*}
\alpha d\left(v\right)&=\alpha \left(k-1+\frac{n}{k}-1\right)=\alpha\left(\sqrt{\frac{\alpha}{1-\alpha}n}+\sqrt{\frac{1-\alpha}{\alpha}n}-2\right)\\ &= \sqrt{\frac{\alpha}{1-\alpha}}\left(\alpha\sqrt{n}+\left(1-\alpha\right)\sqrt{n}\right)-2\alpha\stackbin{\alpha\ge 1/2}{\le} \sqrt{\frac{\alpha}{1-\alpha}n}-1=k-1.
\end{align*}
\paragraph{Upper Bounds.} The trivial upper bound of $\overleftarrow{MD}(G,\alpha)\le n$ is tight for $\alpha>1/2$. Consider the $\alpha$-threshold model for $\alpha>1/2$ on the cycle $C_n$. We claim that a dynamo must include all nodes. Let $\mathcal{C}$ be a configuration on $C_n$ with at least one white node, in the next round both its neighbors will be white. Thus, a configuration with one or more white nodes never reaches a fully black configuration. For the case of $\alpha \le 1/2$ and the majority model, we do not know whether this trivial upper bound is tight or not.
\paragraph{Monotone Dynamo.}Similar to dynamo, one might bound the minimum size of a monotone dynamo. Let us start with the upper bounds. For the $\alpha$-threshold model and $\alpha>1/2$, the trivial upper bound of $n$ is tight for the cycle $C_n$ and there is not much known about the case of $\alpha\le 1/2$ and the majority model. On the other hand, any two adjacent nodes in $C_n$ are a monotone dynamo in the $\alpha$-threshold model for $\alpha \le 1/2$; thus, the trivial lower bound of 2 is tight. For the majority model, Peleg~\cite{peleg1998size} proved that $\sqrt{n}-1\le \overleftarrow{MD}_{mon}(G,maj)$. We prove that $\sqrt{\alpha n/(1-\alpha)}-1\le \overleftarrow{MD}_{mon}(G,\alpha)$ for $\alpha>1/2$ in Theorem~\ref{general:thm6}. Since the dynamos given in the constructions described above are actually monotone dynamos, both these bounds are tight, up to an additive constant. 
\begin{theorem}
\label{general:thm6}
For a graph $G=(V,E)$ and $\alpha>1/2$, we have $\sqrt{\alpha n/(1-\alpha)}-1\le \overleftarrow{MD}_{mon}(G,\alpha)$.
\end{theorem}
\begin{proof}
Let $D\subseteq V$ be a monotone dynamo in the $\alpha$-threshold model on $G$. Suppose the process starts from the configuration where only all nodes in $D$ are black. Recall that $B_t$ denotes the set of black nodes in round $t$. Then, $B_0=D$ and $B_t\subseteq B_{t+1}$ by the monotonicity of $D$. Furthermore, define the potential function $\Phi_t:=\partial\left(B_t\right)$. We claim that $\Phi_{t}\le \Phi_{t-1}-|B_{t}\setminus B_{t-1}|$ because for any newly added black node (i.e., any node in $B_{t}\setminus B_{t-1}$) the number of neighbors in $B_{t-1}$ is strictly larger than $V\setminus B_{t}$ (note that $\alpha>\frac{1}{2}$). In addition, since $D$ is a dynamo, $\mathcal{C}_T|_V=b$ for some $T\ge 0$, which implies $\Phi_T=0$. Thus,
\[
\Phi_T=0\le \Phi_0-\left(n-|D|\right).
\]
This yields
\begin{equation}
\label{general:eq 1}
n\le \Phi_0+|D|.
\end{equation}
For $v\in D$, $d_{V\setminus D}\left(v\right)\le \frac{1-\alpha}{\alpha}d_{D}\left(v\right)$ because $D$ is a monotone dynamo and at least $\alpha$ fraction of $v$'s neighbors must be in $D$. Furthermore, $d_{D}\left(v\right)\le |D|-1$, which implies that $$d_{V\setminus D}\left(v\right)\le \frac{1-\alpha}{\alpha}\left(|D|-1\right).$$Now, we have
\begin{multline}
\label{general:eq 2}
\Phi_0=\partial\left(D\right)=\sum_{v\in D}d_{V\setminus D}\left(v\right)\le \frac{1-\alpha}{\alpha}\sum_{v\in D} \left(|D|-1\right)=\frac{1-\alpha}{\alpha}|D|^2-\frac{1-\alpha}{\alpha}|D|.
\end{multline}
Putting Equations~(\ref{general:eq 1}) and ~(\ref{general:eq 2}) together implies that
\[
n\le\frac{1-\alpha}{\alpha}|D|^2+\left(1-\frac{1-\alpha}{\alpha}\right)|D|.
\]
Some small calculations yields that $\sqrt{\frac{\alpha n}{1-\alpha}}-1\le |D|$. \qed
\end{proof}
\subsubsection{r-Monotone Model}
Trivially, we have $r\le MD(G,r)$ and this is tight for the complete graph $K_n$. For the upper bound, it is proven~\cite{reichman2012new} that $MD(G,r)\le rn/\left(r+1\right)$. 
\subsubsection{r-Threshold Model}
We trivially have $r\le \overleftarrow{MD}\left(G,r\right)\le n$. These bounds are tight for $r\ge 2$ because $\overleftarrow{MD}\left(K_n,r\right)=r$ and $\overleftarrow{MD}\left(G,r\right)=n$ for any $r$-regular graph $G$ (see Lemma~\ref{general:lemma}). For $r=1$, we prove in Theorem~\ref{general:thm4} that $\overleftarrow{MD}\left(G,r\right)$ is equal to 1 or 2.
\begin{theorem}
\label{general:thm4}
In a graph $G$, for $r=1$ $\overleftarrow{MD}\left(G,r\right)=2$ if $G$ is bipartite and $\overleftarrow{MD}\left(G,r\right)=1$ otherwise.
\end{theorem}
\begin{proof}
Let us first argue that any two adjacent nodes in a connected graph $G$ are a dynamo in the $1$-threshold model. Let $v$ and $u$ be two adjacent nodes and assume that the process triggers from a configuration in which $v$ and $u$ are black. A simple inductive argument implies that in the $t$-th round for $t\ge 0$ all nodes in distance at most $t$ from $v$ (similarly $u$) are black. Thus, after $t'$ rounds for some $t'$ smaller than the diameter of $G$ the graph is fully black. 

Now, we prove that $\overleftarrow{MD}\left(G,r\right)=2$ for $r=1$ if $G$ is bipartite. From above, we know that $\overleftarrow{MD}\left(G,r\right)\le 2$; thus, it remains to show that $\overleftarrow{MD}\left(G,r\right)\ge 2$ in this setting. We argue that a configuration with only one black node cannot make $G$ fully black. Since $G$ is bipartite, we can partition the node set of $G$ into two non-empty independent sets $V_1$ and $V_2$. Without loss of generality, assume that we start from a configuration where a node in $V_1$ is black and all other nodes are white. Since all neighbors of the nodes in $V_1$ are in $V_2$ and vice versa, the color of each node in $V_1$ in round $t$ is only a function of the color of nodes in $V_2$ in the $\left(t-1\right)$-th round and the other way around. This implies that by starting from such a configuration, in the next round all nodes in $V_1$ will be white because all nodes in $V_2$ are white initially. In the round after that, all nodes in $V_2$ will be white with the same argument and so on. By an inductive argument, $\mathcal{C}_t|_{V_1}=w$ for odd $t$ and $\mathcal{C}_t|_{V_2}=w$ for even $t$. Thus, there is no dynamo of size one.

Finally, we prove that if $G$ is non-bipartite, then it has a dynamo of size $1$. Since $G$ is not bipartite, it has at least one odd cycle. Let $C_n:=v_1,v_2,\cdots,v_{2k+1},v_1$ be an arbitrary odd cycle in $G$ of size $2k+1$ for some integer $k\ge 1$. Now, suppose that the process starts from a configuration where node $v_1$ is black. In the next round nodes $v_2$ and $v_{2k+1}$ will be black. After one more round nodes $v_3$ and $v_{2k}$ will be black. By applying the same argument, after $k$ rounds nodes $v_{k+1}$ and $v_{k+2}$ will be black. As we discussed above, two adjacent black nodes make the whole graph black. Thus, in a non-bipartite graph, a node which is on an odd cycle is a dynamo of size one in the $1$-threshold model. \qed
\end{proof}
\paragraph{Monotone Dynamo.}For a graph $G=\left(V,E\right)$, the minimum size of a monotone dynamo in the $r$-threshold model is lower-bounded by $r+1$. To prove this, assume that there is a monotone dynamo $D$ of size $r$ or smaller. If $\mathcal{C}_0|_D=b$ and $\mathcal{C}_0|_{V\setminus D}=w$, then $\mathcal{C}_1|_D=w$; this is in contradiction with the monotonicity of $D$. This lower bound is tight because in the complete graph $K_n$, a set of size $r+1$ is a monotone dynamo. Furthermore, the trivial upper bound of $n$ is tight for $r$-regular graphs with $r\ge 2$ (see Lemma~\ref{general:lemma}). For $r=1$, any two adjacent nodes are a monotone dynamo in the $r$-threshold model; thus, $\overleftarrow{MD}_{mon}(G,1)=2$.

\subsection{Number of Edges}
\label{general:nymberofedges}
In this section, we aim to understand how the edge density of the underlying graph influences the minimum size of a dynamo in the $r$-monotone and $r$-threshold model. First we present some bounds on $MD(G,r)$ and $\overleftarrow{MD}(G,r)$ as a function of $m$, $n$, and $r$. Then, we provide sufficient condition, in terms of the minimum degree, for a graph $G$ to have a dynamo of size $r$ in the $r$-threshold model.

There is a simple argument to show that $(n-\frac{m}{r})\le MD(G,r)$ for any graph $G=(V,E)$. Consider the $r$-monotone model on $G$. Let set $D\subset V$ be a dynamo. Assume that initially only nodes in $D$ are black. Recall that $B_t$ for $t\ge 0$ denotes the set of black nodes in $\mathcal{C}_t$. Clearly, each node in $B_t$ for $t\ge 1$ has at least $r$ neighbors in $\bigcup_{t'=0}^{t-1}B_{t'}$. Therefore, $m\ge \left(n-|D|\right)r$, which implies that $|D|\ge \left(n-\frac{m}{r}\right)$. 

Now, we prove in Lemma~\ref{bipartite} that $2MD(G,r)\le\overleftarrow{MD}(G,r)$ if $G$ is bipartite. Therefore, we can conclude that $2\left(n-\frac{m}{r}\right)\le \overleftarrow{MD}(G,r)$ for any bipartite graph $G$.
\begin{lemma}\label{bipartite}
Let $G=\left(V_1\cup V_2,E\right)$ be a bipartite graph. Then $2MD\left(G,r\right)\le \overleftarrow{MD}\left(G,r\right)$.
\end{lemma}
\begin{proof} 
Let $D$ be a dynamo in the $r$-threshold model on $G$. We construct a dynamo of size at most $|D|/2$ in the $r$-monotone model. Define $D_1:=V_1\cap D$, and without loss of generality assume that $|D_1|\le |D|/2$. We claim that $D_1$ is a dynamo in the $r$-monotone model on $G$.

For a node $v$ and $t\ge 0$, let $\mathcal{C}_t\left(v\right)$ (resp. $\mathcal{C}'_t\left(v\right)$) denote $v$'s color in the $t$-th round of the $r$-threshold model (resp. the $r$-monotone model) on $G$ assuming that initially only nodes in set $D$ (resp. $D_1$) are black. We claim that for any node $v\in V_1$ and $t\ge 0$, if $\mathcal{C}_{2t}\left(v\right)=b$ then $\mathcal{C}'_{2t}\left(v\right)=b$. We prove this claim by applying induction. The base case $t=0$ is true by definition. For some $t\ge 1$ and $v\in V_1$, assume that $\mathcal{C}_{2t}\left(v\right)=b$. This implies that $v$ has $r$ neighbors $u_1,\cdots, u_r$ in $V_2$ such that $\mathcal{C}_{2t-1}\left(u_i\right)=b$ for $1\le i\le r$. Therefore, each $u_i$ has $r$ neighbors $u_{i}^{1},\cdots, u_i^{r}$ in $V_1$ so that $\mathcal{C}_{2t-2}(u_{i}^{j})=b$ for $1\le j\le r$. By the induction hypothesis, $\mathcal{C}'_{2t-2}(u_i^{j})=b$ for any $1\le i,j\le r$. Since in the $r$-monotone model a node with $r$ black neighbors becomes black, we have $\mathcal{C}'_{2t-1}\left(u_i\right)=b$ for $1\le i\le r$. Therefore, $\mathcal{C}'_{2t}\left(v\right)=b$. Let $t_f$ be the minimum $t$ for which $\mathcal{C}_{2t_f}|_{V_1}=b$ (note that such $t$ exists since $D$ is a dynamo). Then, we have $\mathcal{C}'_{2t_f}|_{V_1}=b$, which implies that $\mathcal{C}'_{2t_f+1}|_V=b$. This is true since in the $r$-monotone model all nodes in $V_1$ remain black in round $2t_f+1$ and all nodes in $V_2$ will become black (note that all neighbors of a node in $V_2$ are in $V_1$, which are all black). Thus, $D_1$ is a dynamo in the $r$-monotone model on $G$. \qed
\end{proof}
  
So far, we showed that $2\left(n-\frac{m}{r}\right)\le \overleftarrow{MD}(G,r)$ for any bipartite graph $G$. Does this bound actually hold for any graph? The answer is positive. We prove such statement in Theorem~\ref{theorem-threshold-lower-bound}. We basically prove that if we have a lower bound of a particular form on the minimum size of a dynamo in the $r$-threshold model for all bipartite graphs, then the same lower bound holds for all graphs.
\begin{theorem}
\label{theorem-threshold-lower-bound}
$2\left(n-\frac{m}{r}\right)\le \overleftarrow{MD}(G,r)$ for any graph $G=(V,E)$. 
\end{theorem}
\begin{proof}
We know that $2\left(n-\frac{m}{r}\right)\le \overleftarrow{MD}(G,r)$ for any bipartite graph $G$ by Lemma~\ref{bipartite}. We want to show that this is actually true for any graph. For the sake of contradiction, assume that there is a graph $H=(V_{H},E_{H})$ which has a dynamo $D_H$ of size $|D_H|<2\left(n_{\scaleto{H}{4pt}}-\frac{m_{\scaleto{H}{4pt}}}{r}\right)$ in the $r$-threshold model, where $n_{\scaleto{H}{4pt}}:=|V_H|$ and $m_{\scaleto{H}{4pt}}:=|E_H|$. Suppose that $V_{H}:=\{v_1,\cdots, v_{n_{\scaleto{H}{4pt}}}\}$. We construct a bipartite graph $F=(V_F,E_F)$, where $V_F=\{x_1,\cdots,x_{n_{\scaleto{H}{4pt}}}\}\cup\{y_1,\cdots, y_{n_{\scaleto{H}{4pt}}}\}$ and we add an edge between $x_i$ and $y_j$ if and only if $\{v_i,v_j\}\in E_H$. Define the node set $D_F$ by including nodes $x_i$ and $y_i$ if and only if $v_i\in D_H$, which implies that $|D_F|=2|D_H|$. We claim that $D_F$ is a dynamo in the $r$-threshold model on $F$. Therefore, the bipartite graph $F$ has a dynamo of size
\[
|D_F|=2|D_H|< 2\cdot 2\left(n_{\scaleto{H}{4pt}}-\frac{m_{\scaleto{H}{4pt}}}{r}\right)=2\left(n_{\scaleto{F}{4pt}}-\frac{m_{\scaleto{F}{4pt}}}{r}\right)
\]  
where we used $n_{\scaleto{F}{4pt}}=2n_{\scaleto{H}{4pt}}$ and $m_{\scaleto{F}{4pt}}=2m_{\scaleto{H}{4pt}}$. This is a contradiction.

It remains to prove that $D_F$ is a dynamo in $F$. Consider the $r$-threshold model on $F$ and $H$ where initially $x_i$ and $y_i$ have the same color as $v_i$. By an inductive argument, it is easy to see that after $t$ rounds for any $t\ge 0$, $x_i$'s and $y_i$'s color will be identical to $v_i$'s color. If initially all nodes in $D_H$ are black, eventually all nodes in $H$ will become black since $D_H$ is a dynamo. Therefore, if initially all nodes in $D_F$ are black, eventually all nodes in $F$ will become black; that is, $D_F$ is a dynamo in $F$. \qed
\end{proof}
\paragraph{Tightness.}Let consider the above bounds for the $d$-dimensional torus $\mathbb{T}_L^d$ which is the graph whose node set is equal to $[L]^d:=\{1,\cdots,L\}^d$ and two nodes are adjacent if and only if they differ by $1$ or $L-1$ in exactly one coordinate. The number of nodes and edges in $\mathbb{T}_L^d$ are equal to $L^d$ and $dL^d$, respectively. Thus, by the aforementioned lower bounds we have
\[
\left(1-\frac{d}{r}\right)L^d\le MD\left(\mathbb{T}_L^d,r\right)\ \ \text{and} \ \ 2\left(1-\frac{d}{r}\right)L^d\le \overleftarrow{MD}\left(\mathbb{T}_L^d,r\right).
\] 
Balister, Bollob\'{a}s, Johnson, and Walters~\cite{balbol09} proved that there is a dynamo of size $\left(1-\frac{d}{r}\right)L^d+\mathcal{O}(L^{d-1})$ in the $r$-monotone model on $\mathbb{T}_L^d$. Furthermore, Jeger and Zehmakan~\cite{jeger2019dynamicc} showed that there exists a dynamo of size $2\left(1-\frac{d}{r}\right)L^d+\mathcal{O}(L^{d-1})$ in the $r$-threshold model on $\mathbb{T}_L^d$. Therefore, the above bounds are tight, up to the terms of lower orders.

For two graphs $G=\left(V,E\right)$ and $G'=\left(V,E'\right)$, if $E\subset E'$ then $MD\left(G',r\right)\le MD\left(G,r\right)$ and $\overleftarrow{MD}\left(G',r\right)\le \overleftarrow{MD}\left(G,r\right)$. This is true because by a simple inductive argument any dynamo in $G$ is also a dynamo in $G'$. Thus, if we keep adding edges to any graph, eventually it will have a dynamo of minimum possible size, namely $r$, in both the $r$-monotone and $r$-threshold model. Thus, it would be interesting to ask for the degree-based density conditions that ensure that a graph has a dynamo of size $r$. Gunderson~\cite{gunderson2017minimum} proved that if $\delta \ge n/2+r$ for a graph $G$ ($r$ can be replaced by $r-3$ for $r\ge 4$), then $MD\left(G,r\right)=r$. We provide similar results for the threshold variant.
\begin{theorem}
\label{general:thm5}
If $\delta\ge n/2+r$ for a graph $G=\left(V,E\right)$, then it has $\Omega\left(n^r\right)$ dynamos of size $r$ in the $r$-threshold model. 
\end{theorem}
Note that this statement is stronger than Gunderson's result in two ways. Firstly, we prove that there is a dynamo of size $r$ in the $r$-threshold model, which immediately implies that there is a dynamo of such size in the $r$-monotone model. Furthermore, we prove that actually there exist $\Omega\left(n^r\right)$ of such dynamos (this is asymptotically best possible since there are ${n\choose r}=\mathcal{O}\left(n^r\right)$ sets of size $r$). It is worth to mention that our proof is substantially shorter.

\begin{proof}
For a subset $D\subseteq V$, define the sets 
\[
V_1^{D}:=\{v\in V:d_D\left(v\right)\ge r\}\ \ \text{and}\ \ V_2^{D}:=V\setminus V_1^{D}.
\]
Firstly, we have
\begin{equation}
\label{general:eq 4}
\sum_{v\in V}d_D\left(v\right)=\sum_{v\in D}d\left(v\right)\ge |D|\delta.
\end{equation}
Furthermore, since each node in $V_1^D$ has at most $|D|$ neighbors in $D$ and each node in $V_2^D$ has at most $\left(r-1\right)$ neighbors in $D$, we have
\[
\sum_{v\in V}d_D\left(v\right)\le |D|\cdot |V_1^D|+\left(r-1\right)|V_2^D|.
\]
By applying $|V_1^D|+|V_2^D|=n$, we get
\begin{equation}
\label{general:eq 5dynamo} \sum_{v\in V}d_D\left(v\right)\le\left(|D|-\left(r-1\right)\right)|V_1^D|+\left(r-1\right)n.
\end{equation}
By combining Equations~(\ref{general:eq 4}) and~(\ref{general:eq 5dynamo}), we have 
\[
\left(|D|-\left(r-1\right)\right)|V_1^D|+\left(r-1\right)n\ge |D|\delta.
\]
Applying $\delta\ge \frac{n}{2}+r$ gives us
\[
\left(|D|-\left(r-1\right)\right)|V_1^D|+\left(r-1\right)n\ge \frac{n}{2}|D|+r|D|.
\]
Dividing both sides by $|D|$ and rearranging the terms yields
\[
|V_1^D|\ge \frac{\frac{1}{2}-\frac{r-1}{|D|}}{1-\frac{r-1}{|D|}}n+\frac{r}{1-\frac{r-1}{|D|}}.
\]
Now, by applying $1-\frac{r-1}{|D|}\le 1$ we get
\begin{equation}
\label{general:eq 3}
|V_1^D|\ge \frac{|D|-2r+2}{2|D|-2r+2}n+r.
\end{equation}
Building on Equation~(\ref{general:eq 3}), we prove that any node set of size $2r-1$ has at least one subset of size $r$ which is a dynamo in the $r$-threshold model. There are ${n\choose 2r-1}=\Omega\left(n^{2r-1}\right)$ sets of size $2r-1$ and a set of size $r$ is shared by ${n-r\choose r-1}=\mathcal
{O}\left(n^{r-1}\right)$ sets of size $2r-1$. Thus, there exist $\Omega\left(n^{2r-1}\right)/\mathcal{O}\left(n^{r-1}\right)=\Omega\left(n^r\right)$ distinct dynamos of size $r$.

Let $D_1\subset V$ be an arbitrary set of size $2r-1$. By setting $D=D_1$ in Equation~(\ref{general:eq 3}) and applying $|D_1|=2r-1$, we have $|V_1^{D_1}|\ge n/2r$. Recall that $V_1^{D_1}$ are the nodes which have at least $r$ neighbors in $D_1$. By an averaging argument, there is a subset $D_2\subset D_1$ of size $r$ so that at least $\left(n/2r\right)/{2r-1\choose r}\ge n/\left(2r\right)^{2r}$ of nodes in $V_{1}^{D_1}$ each has $r$ neighbors in $D_2$. We want to prove that $D_2$ is a dynamo in the $r$-threshold model. If initially $D_2$ is fully black, in the next round at least $n/\left(2r\right)^{2r}$ nodes will be black. Now, we prove that if an arbitrary set $D_3$ of size at least $n/\left(2r\right)^{2r}$ is black, the whole graph becomes black in at most three more rounds.

By setting $D=D_3$ in Equation~(\ref{general:eq 3}), we have 
\begin{align*}
|V_1^{D_3}|\ge \frac{|D_3|-2r+2}{2|D_3|}n \ge\frac{n}{2}-\frac{\left(r-1\right)\left(2r\right)^{2r}}{n}n \ge \frac{n}{2}-\left(2r\right)^{2r+1} \ge \frac{n}{2}-\left(2r\right)^{3r}
\end{align*}
which is the number of black nodes generated by $D_3$. Let $D_4$ be a set of size at least $n/2-\left(2r\right)^{3r}$. We show that $|V_1^{D_4}|\ge n/2$, which implies that if $\mathcal{C}_t|_{D_4}=b$ for some $t\ge 0$, there will be at least $\frac{n}{2}$ black nodes in $\mathcal{C}_{t+1}$. Again applying Equation~(\ref{general:eq 3}) gives us
\begin{align*}
|V_1^{D_4}|&\ge \frac{|D_4|-2r+2}{2|D_4|-2r+2}n+r\ge\frac{\frac{n}{2}-\left(2r\right)^{3r}-2r+2}{n-2\left(2r\right)^{3r}-2r+2}n+r\ge\frac{n}{2}.
\end{align*}
The last inequality follows from some straightforward calculations (omitted).

Finally, we prove that if $n/2$ nodes are black in some configuration, the whole graph becomes black in the next round. Consider an arbitrary node $v$. Since $d\left(v\right)\ge \frac{n}{2}+r$ and there are at least $n/2$ black nodes, $v$ has at least $r$ black neighbors and will become black in the next round. \qed
\end{proof}
\section{Robust and Eternal Sets}
\label{general:robust}
In this section, we provide tight bounds on the minimum size of a robust set and an eternal set as a function of $n$, $r$, and $\alpha$. See Tables~\ref{general:Table 2Robust} and~\ref{general:Table 2Eternal} for a summary. 

In the $\alpha$-monotone and $r$-monotone model a black node remains unchanged, which implies that $MR\left(G,\alpha\right)=ME\left(G,\alpha\right)=MR\left(G,r\right)=ME\left(G,r\right)=1$ for a graph $G$. Thus, we focus on the $\alpha$-threshold, $r$-threshold, and majority model in the rest of the section.
\subsection{$\mathbf{\alpha}$-Threshold Model}
\subsubsection{Robust Sets}
We present tight bounds on $\overleftarrow{MR}\left(G,\alpha\right)$ in Theorem~\ref{general:thm7}.
\begin{theorem}
\label{general:thm7}
For a graph $G=\left(V,E\right)$,
\begin{itemize}
\item[(i)]$\lceil\frac{1}{1-\alpha}\rceil\le\overleftarrow{MR}\left(G,\alpha\right)\le n$ for $\alpha>1/2$
\item[(ii)]$2=\lceil\frac{1}{1-\alpha}\rceil\le\overleftarrow{MR}\left(G,\alpha\right)\le 2\alpha n+1/\alpha$ for $\alpha\le1/2$.
\end{itemize}
\end{theorem}

\begin{proof}
To prove the lower bound of $\lceil\frac{1}{1-\alpha}\rceil$, let the node set $S\subseteq V$ be a robust set in the $\alpha$-threshold model. Since $G$ is connected, there is a node $v\in S$ that shares at least one edge with $V\setminus S$, which implies that $d\left(v\right)\ge d_S\left(v\right)+1$. Furthermore, $d_S\left(v\right)\ge \alpha d\left(v\right)$ because $S$ is robust. Hence, $d_S\left(v\right)\ge \alpha\left(d_S\left(v\right)+1\right)$ which yields $d_S\left(v\right)\ge \alpha/(1-\alpha)$. Moreover, $|S|-1\ge d_S\left(v\right)$, which implies that $|S|\ge \alpha/(1-\alpha)+1=1/(1-\alpha)$. As $|S|$ is a positive integer, we have $|S|\ge\lceil1/(1-\alpha)\rceil$.

Now, let us prove the upper bound of $2\alpha n+1/\alpha$. We prove that in the $\alpha$-threshold model on $G$, there is a robust set of size at most $2\alpha n+1/\alpha$. Assume that $\mathcal{P}$ is the set of all partitions of $V$ into $\lfloor 1/\alpha\rfloor$ sets such that all sets are of size at least $\lfloor \alpha n\rfloor$, except one set which can be of size $\lfloor \alpha n\rfloor-1$. Let $P\in \mathcal{P}$ be a partition for which the number of edges between the sets is minimized. Let $V_{\max}$ be a set of maximum size in $P$. Clearly, $V_{\max}$ is at least of size $\alpha n$ and at most of size
\begin{align*}
|V_{\max}|&\le n-\left(\lfloor \frac{1}{\alpha}\rfloor-2\right)\lfloor \alpha n\rfloor-\left(\lfloor\alpha n\rfloor-1\right)\\ &=n-\lfloor \frac{1}{\alpha}\rfloor\lfloor \alpha n\rfloor+\lfloor \alpha n\rfloor+1\\ &\le n-\left(\frac{1}{\alpha}-1\right)\left(\alpha n-1\right)+\alpha n+1\\ &=2\alpha n+1/\alpha.
\end{align*}  
Furthermore, we claim that for each node $v\in V_{\text{max}}$, $d_{V_{\text{max}}}\left(v\right)\ge \alpha d(v)$, which implies that $V_{\max}$ is a robust set of size at most $2\alpha n+1/\alpha$. Assume that there is a node $u$ which violates this property, i.e., $d_{V_{\max}}\left(u\right)<\alpha d(u)$. Then, the average number of edges between $u$ and the $\lfloor 1/\alpha\rfloor-1$ other sets is at least
\[
\frac{\left(1-\alpha\right)d\left(u\right)}{\lfloor \frac{1}{\alpha}\rfloor-1}\ge \frac{\left(1-\alpha\right)d\left(u\right)}{\frac{1}{\alpha}-1}=\alpha d\left(u\right).
\] 
Thus, there must exist a set $V'$ among the other $\lfloor 1/\alpha\rfloor-1$ sets such that $d_{V'}\left(u\right)\ge \alpha d(u)>d_{V_{\max}}\left(u\right)$. This is a contradiction because by removing $u$ from $V_{\max}$ and adding it into $V'$, the number of edges between the sets decreases at least by one. Furthermore, it is easy to see that the new partition is also in $\mathcal{P}$ by using $V_{\max}\ge \alpha n$. \qed
\end{proof}

\textbf{Tightness.} The lower bounds in Theorem~\ref{general:thm7} are tight. Consider graph $G=\left(V,E\right)$ with $V=V_1\cup V_2$, where $|V_1|=n-\lceil 1/(1-\alpha)\rceil$ and $|V_2|=\lceil 1/(1-\alpha)\rceil$. Assume that the nodes in $V_2$ induce a clique, the nodes in $V_1$ induce an arbitrary connected graph, and finally there is an edge between a node $v\in V_2$ and a node $v'\in V_1$. Graph $G$ is connected and has $n$ nodes, by construction. Moreover, $V_2$ is a robust set for $G$ in the $\alpha$-threshold model because firstly for each node $u\in V_2\setminus\{v\}$, $d_{V_2}\left(u\right)=d\left(u\right)\ge \alpha d\left(u\right)$. Furthermore, for node $v$, we have
\begin{align*}
d_{V_2}\left(v\right)=\lceil \frac{1}{1-\alpha}\rceil-1\ge \lceil \frac{1}{1-\alpha}\rceil-(1-\alpha)\lceil\frac{1}{1-\alpha}\rceil=\alpha \lceil \frac{1}{1-\alpha}\rceil=\alpha d\left(v\right).
\end{align*} 
where we used $d\left(v\right)=\lceil \frac{1}{1-\alpha}\rceil$. 

The upper bound of $\overleftarrow{MR}(G,\alpha)\le n$ is tight since $\overleftarrow{MR}\left(C_n,\alpha\right)=n$ for $\alpha>\frac{1}{2}$. We do not know whether the upper bound of $\overleftarrow{MR}(G,\alpha)\le 2\alpha n+1/\alpha$ for $\alpha\le 1/2$ is tight or not.
\subsubsection{Eternal Sets}
For a graph $G$, $1\le\overleftarrow{ME}\left(G,\alpha\right)\le n$ for $\alpha>1/2$ and $1\le\overleftarrow{ME}\left(G,\alpha\right)\le 2\alpha n+1/\alpha$ for $\alpha\le1/2$. All bounds are trivial except $2\alpha n+1/\alpha$, which is a corollary of Theorem~\ref{general:thm7}. (Note that any robust set is also an eternal set) The lower bound of 1 is tight since for the star graph $S_n$ $\overleftarrow{ME}(S_n,\alpha)=1$. The upper bound of $n$ for $\alpha>1/2$ is tight because $\overleftarrow{ME}\left(C_n,\alpha\right)=n$ for $\alpha>1/2$ and odd $n$ (this basically follows from the proof of Lemma~\ref{general:obs1} by replacing black with white and $\alpha\le1/2$ with $\alpha>1/2$).
\subsection{Majority Model}
We want to bound the minimum size of a robust set and an eternal set in the majority model on a graph $G$.
\paragraph{Lower Bounds.}The trivial lower bound of $1\le \overleftarrow{ME}(G,maj)$ is tight for the star graph $S_n$ since the internal node is an eternal set. Any robust is of size at least 2 in this setting and in the cycle $C_n$ two adjacent nodes are a robust set of size 2. 
\paragraph{Upper Bounds.}In Theorem~\ref{general:majoritythm1}, we prove that $\overleftarrow{MR}(G,maj)\le \lfloor n/2\rfloor+1$ which implies that $\overleftarrow{ME}(G,maj)\le \lfloor n/2\rfloor+1$. These upper bounds are tight, up to an additive constant for the complete graph $K_n$.
\begin{theorem}
\label{general:majoritythm1}
For any graph $G=(V,E)$, $\overleftarrow{MR}(G,maj)\le \lfloor n/2\rfloor+1$. 
\end{theorem}
\begin{proof}
The proof is similar to the proof of the upper bound of $\overleftarrow{MR}(G,\alpha)\le 2\alpha n+1/\alpha$ in Theorem~\ref{general:thm7}. We partition $V$ into two sets $V_1$ and $V_2$ such that $$\lfloor n/2\rfloor-1\le |V_2|\le|V_1|\le \lfloor n/2\rfloor+1$$ and the number of edges in between is minimized. It is sufficient to prove that $V_1$ is a robust set. Thus, assume otherwise, i.e., there is a node $v\in V_1$ such that $d_{V_1}(v)<d_{V_2}(v)$. In that case, we define $V'_1:=V_1\setminus\{v\}$ and $V'_2:=V_2\cup\{v\}$. We observe that $$\lfloor n/2\rfloor-1\le |V'_1|\le|V'_2|\le \lfloor n/2\rfloor+1$$ and the number of edges between $V'_1$ and $V'_2$ is strictly less than the number of edges between $V_1$ and $V_2$. This is a contradiction. \qed
\end{proof}
\subsection{r-Threshold Model}
\subsubsection{Robust Sets}
For a graph $G$, $\overleftarrow{MR}\left(G,r\right)=2$ for $r=1$ because two adjacent nodes create a robust set. For $r\ge 2$, we have the tight bounds of $r+1\le\overleftarrow{MR}\left(G,r\right)\le n$. Notice that in the $r$-threshold model if in a configuration less than $r+1$ nodes are black, in the next round all black nodes turn white. Furthermore, the lower bound of $r+1$ is tight for $K_n$ and the upper bound is tight for $r$-regular graphs.
\subsubsection{Eternal Sets}
We provide tight bounds on the minimum size of an eternal set in the $r$-threshold model on a graph $G$ in Theorem~\ref{general:thm8}. 
\begin{theorem}
\label{general:thm8}
For a graph $G=\left(V,E\right)$ 
\begin{itemize}
\item[(i)]if $r=1$, $\overleftarrow{ME}\left(G,r\right)=1$
\item[(ii)]if $r=2$, $2\le \overleftarrow{ME}\left(G,r\right)\le \frac{n}{1+x}$, where $x=0$ for odd $n$ and $x=1$ for even $n$
\item[(iii)]if $r\ge 3$, $r\le \overleftarrow{ME}\left(G,r\right)\le n$.
\end{itemize}
\end{theorem}
\begin{proof}
For $r=1$, any node set of size 1 is eternal because if a node $v\in V$ is black in some configuration in the $1$-threshold model, in the next round all nodes in $\Gamma\left(v\right)$ will be black.

The lower bound of $r$ is trivial because a configuration with less than $r$ black nodes in the $r$-threshold model becomes fully white in the next round. 

All the upper bounds are also trivial, except the bound of $n/2$. Assume that $n$ is even and $r=2$; we prove that $G$ has an eternal set of size at most $n/2$. Let
$C_k:=u_1,u_2,\cdots, u_i,u_{i+1},\cdots, u_k,u_1$ 
be a cycle of length $k$ in $G$, then the node set $U:=\{u_1,\cdots,u_k\}$ is a robust set, and consequently an eternal set in the $2$-threshold model. If $k$ is even, then actually the smaller set $U_{\text{e}}:=\{u_i\in U:i\ \text{is even} \}$ is an eternal set. This is correct because in the $2$-threshold model if for some configuration all nodes in $U_{\text{e}}$ are black, in the next round all nodes in $U\setminus U_{\text{e}}$ will be black, irrespective of the color of the other nodes. One round later, all nodes in $U_{\text{e}}$ will be black again and so on. 

Now, let $C$ of length $k$ be a longest cycle in $G$. If $k\le n/2$, then the set of nodes in $C$ is an eternal set of size at most $n/2$. If $k$ is larger than $n/2$ but it is even, then half of the nodes in $C$ suffice to form an eternal set of size at most $n/2$. Thus we are left with the case that $k$ is odd and $k>n/2$. Let the node set $V_1$ include all nodes in $C$ and $V_2:=V\setminus V_1$. Since the graph is connected and $V_2\ne \emptyset$ (we already have excluded the case of $k= n$ sine $n$ is even), there is a node $y_1\in V_2$ and a node $u\in V_1$ such that $\{y_1,u\}\in E$. We start traversing from $y_1$. Since $r=2$, each node, including $y_1$, is of degree at least $2$ (recall that we always assume $\delta\left(G\right)\ge r$). Thus, $y_1$ in addition to $u$ has at least another neighbor, say $y_2$. Since $y_2$ is of degree at least $2$ as well, it must be adjacent to another node, say $y_3$, and so on. Assume that $y_{k'}$ for some $k'\ge 2$ is the first node from $V_1$, which we visit during the traversing. The union of path $u, y_1, y_2,\cdots, y_{k'}$ and cycle $C$ gives us two new cycles if $u\ne y_{k'}$. Since the length of cycle $C$ is odd, one of these two new cycles must be of even length (see Fig.~\ref{general:GeneralBoundsFig8}). This even cycle provides us with an eternal set of size $n/2$. If $u=y_{k'}$ then we have a cycle whose nodes are all from $V_2$ plus node $u$. This cycle is of size at most $n/2$ because $|V_2|=|V|-|V_1|=n-k< n/2$ by applying $k>n/2$. This gives an eternal set of size at most $n/2$. If we never visit a node from $V_1$, then we eventually revisit a node $y_j\in V_2$ which gives us a cycle on some nodes in $V_2$. Note this cycle is of size at most $n/2$ since $|V_2|<n/2$.
\begin{figure}[hbt!]
\centering
\includegraphics[scale=0.8]{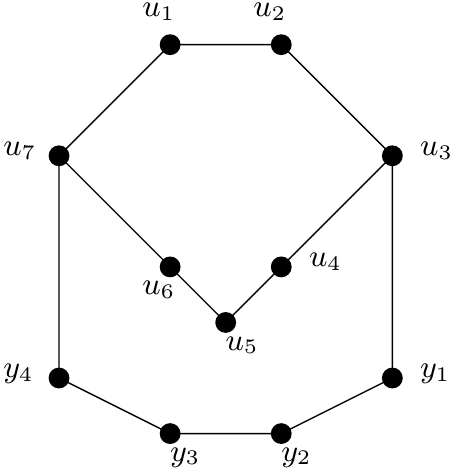}
\caption{By connecting each of the endpoints of a path into a different node on a cycle, two new cycles are generated. If the original cycle is of odd size, here $7$, then one of the two new cycles is of even size, here the cycle $u_1,u_2,u_3,y_1,\cdots,u_7,u_1$.}
\label{general:GeneralBoundsFig8}
\end{figure} \qed
\end{proof}

\textbf{Tightness.} The lower bounds are tight because $\overleftarrow{ME}\left(K_n,r\right)=r$ for $r\ge 1$. The upper bounds for $r=2$ are tight because $\overleftarrow{ME}\left(C_n,2\right)$ is equal to $n$ for odd $n$ and it is equal to $n/2$ for even $n$. 

For the tightness of the upper bound of $\overleftarrow{ME}(G,r)\le n$ for $r\ge 3$, we prove that for any sufficiently large $n$ (regardless of its parity), there is an $n$-node graph $G=\left(V,E\right)$ which has no eternal set of size smaller than $n-2r-1$. Let $n$ be even; we will show how our argument applies to the odd case. We first present the construction of graph $G$ step by step. For $1\le i\le K:=\lfloor n/\left(r+1\right)\rfloor-1$, let $G_i$ be the clique on the node set $V_i:=\{v_i^{\left(j\right)}:1\leq j\leq r+1\}$ minus the edge $\{v_i^{\left(1\right)},v_i^{\left(2\right)}\}$. To create the first part of graph $G$, we connect $G_i$s with a path. More precisely, we add the edge set $\{\{v_i^{\left(2\right)},v_{i+1}^{\left(1\right)}\}: 1\le i\le K-1\}$. So far the generated graph is $r$-regular except the nodes $v_1^{\left(1\right)}$ and $v_K^{\left(2\right)}$ which are of degree $r-1$ and we have $\ell:=n-K\left(r+1\right)$ nodes left. Let $G'$ be an arbitrary $r$-regular graph on $\ell$ nodes. Now, remove an edge $\{v',u'\}$ from $G'$ and connect $v'$ to $v_1^{\left(1\right)}$ and $u'$ to $v_{K}^{\left(2\right)}$. Clearly the resulting graph $G$ is $r$-regular with $n$ nodes (see Figure~\ref{general:GeneralBoundsFig9} for an example). However, we should discuss that such a graph $G'$ exists. The necessary and sufficient condition for the existence of an $r$-regular graph on $\ell$ nodes is that $\ell\ge r+1$ and $r\ell$ is even. Firstly, we have $$\ell=n-\left(\left\lfloor \frac{n}{\left(r+1\right)}\right\rfloor-1\right)\left(r+1\right)\ge n-\left( \frac{n}{\left(r+1\right)}-1\right)\left(r+1\right) =r+1.$$ Furthermore, if $r$ is even, then $r\ell$ is even; thus, assume otherwise. Since $r$ is odd and $n$ is even, then $rn$ is even. In addition, $r\left(r+1\right)$ is even, which implies that $r\left(r+1\right)\left(\lfloor n/\left(r+1\right)\rfloor-1\right)$ is even. Overall, $r\ell=rn-r\left(r+1\right)\left(\lfloor n/\left(r+1\right)\rfloor-1\right)$ is even. 
\begin{figure}[H]
\centering
\includegraphics[scale=0.8]{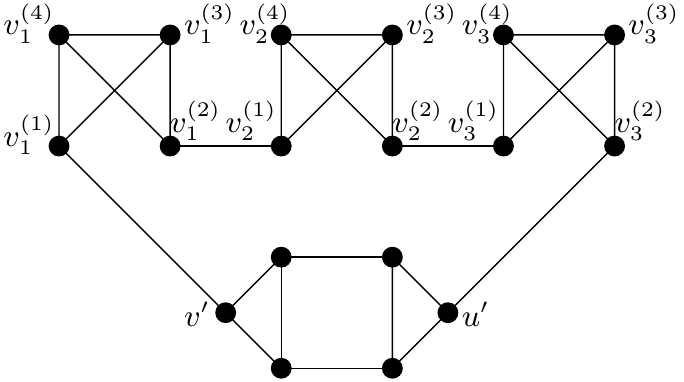}
\caption{The construction of graph $G$ for $r=3$ and $n=18$.}
\label{general:GeneralBoundsFig9}
\end{figure}
Now, we claim that all nodes in $V_i$ for $1\le i\le K$ must be in any eternal set. This implies that the minimum size of an eternal set in $G$ is at least $$K\left(r+1\right)=\left(\left\lfloor \frac{n}{\left(r+1\right)}\right\rfloor-1\right)\left(r+1\right)\ge n-2r-1.$$ Assume that we start from a configuration in which all nodes are black except a node $v$ in $V_i$ for some $1\le i\le K$. If $v$ is $v_i^{\left(1\right)}$ or $v_i^{\left(2\right)}$, then in the next round $v_i^{\left(3\right)}$ and $v_i^{\left(4\right)}$ (which must exist because $r\ge 3$) are both white since each of them has at most $r-1$ black neighbors. By construction, there is an edge between $v_i^{\left(3\right)}$ and $v_i^{\left(4\right)}$. Since $G$ is $r$-regular, in the next round they both stay white and all their neighbors become white as well. By an inductive argument in the $t$-th round for $t\ge 1$, all nodes whose distance is at most $t$ from $v_i^{\left(3\right)}$ (or similarly $v_i^{\left(4\right)}$) will be white. Thus, eventually the whole graph becomes white. Now, assume that $v$ is a node in $V_i\setminus\{v_i^{\left(1\right)},v_i^{\left(2\right)}\}$. Again, in the next round $v_i^{\left(1\right)}$ and one round after that $v_i^{\left(3\right)}$ and $v_i^{\left(4\right)}$ become white and the same argument follows. 

If $n$ is odd, we do the same construction for $n-1$ instead of $n$ and at the end, add a node $w$ and connect it to $v_{i}^{\left(2\right)}$ for $1\le i\le r$. This graph is not $r$-regular since there are $r$ nodes of degree $r+1$. However, a similar argument applies since from any configuration with two adjacent white nodes in one of $V_i$s, white color eventually takes over. 

\bibliographystyle{acm}
\bibliography{refer}
\end{document}